\documentclass[11pt]{article}

\usepackage{style}

\title{Average-Case Integrality Gap for Non-Negative Principal Component Analysis}

\author[1]{Afonso S.\ Bandeira\thanks{Email: \textit{bandeira@math.ethz.ch}. Some of this work was done while with the Department of Mathematics at the Courant Institute of Mathematical Sciences, and the Center for Data Science, at New York University; and partially supported by NSF grants DMS-1712730 and DMS-1719545, and by a grant from the Sloan Foundation.}}
\author[2]{Dmitriy Kunisky\thanks{Email: \textit{kunisky@cims.nyu.edu}. Partially supported by NSF grants DMS-1712730 and DMS-1719545.}}
\author[2]{Alexander S.\ Wein\thanks{Email: \textit{awein@cims.nyu.edu}. Partially supported by NSF grant DMS-1712730 and by the Simons Collaboration on Algorithms and Geometry.}}

\date{December 3, 2020}

\affil[1]{Department of Mathematics, ETH Z\"{u}rich}
\affil[2]{Department of Mathematics, Courant Institute of Mathematical Sciences, NYU}

\definecolor{Orange}{RGB}{220,110,0}

\begin{document}

\maketitle

\abstract{
Montanari and Richard (2015) asked whether a natural semidefinite programming (SDP) relaxation can effectively optimize $\bx^{\top}\bW \bx$ over $\|\bx\| = 1$ with $x_i \geq 0$ for all coordinates $i$, where $\bW \in \RR^{n \times n}$ is drawn from the Gaussian orthogonal ensemble (GOE) or a spiked matrix model.
In small numerical experiments, this SDP appears to be \emph{tight} for the GOE, producing a rank-one optimal matrix solution aligned with the optimal vector $\bx$.
We prove, however, that as $n \to \infty$ the SDP is not tight, and certifies an upper bound asymptotically no better than the simple spectral bound $\lambda_{\max}(\bW)$ on this objective function.
We also provide evidence, using tools from recent literature on hypothesis testing with low-degree polynomials, that no subexponential-time certification algorithm can improve on this behavior.
Finally, we present further numerical experiments estimating how large $n$ would need to be before this limiting behavior becomes evident, providing a cautionary example against extrapolating asymptotics of SDPs in high dimension from their efficacy in small ``laptop scale'' computations.
}

\pagenumbering{gobble}

\newpage

\pagenumbering{arabic}

\noindent
\section{Introduction}

Recovering the most significant directions or \emph{principal components} of a matrix from noisy observations is a fundamental problem in both mathematical statistics and applications \cite{Rao-1964-PCA,Johnstone-2001-LargestEigenvaluePCA,Ringner-2008-PCA,AW-2010-PCA}.
The theoretical asymptotics of this task have been studied at length by analyzing idealized \emph{spiked matrix models}.
While the first such models, proposed by Johnstone \cite{Johnstone-2001-LargestEigenvaluePCA}, concerned Gaussian observations from a covariance matrix deformed by adding a rank-one ``spike,'' the following simpler \emph{Wigner spiked matrix model} captures much of the same phenomenology.
We consider the following two probability distributions $\PP$ and $\QQ$ over $n \times n$ symmetric matrices.
\begin{itemize}
    \item Under $\QQ = \QQ_n$, observe $\bW \in \RR^{n \times n}_{\sym}$ drawn from the \emph{Gaussian orthogonal ensemble (GOE)}, meaning $W_{ii} \sim \sN(0, \frac{1}{n})$ and $W_{ij} = W_{ji} \sim \sN(0, \frac{2}{n})$ for $i < j$ with all $\frac{n(n + 1)}{2}$ of these entries distributed independently. We also write $\QQ = \GOE(n)$ for this distribution.
    \item Under $\PP = \PP_n$, first draw $\bu \sim \Unif(\mathbb{S}^{n - 1})$ for $\mathbb{S}^{n - 1} \subset \RR^n$ the sphere of unit radius, and then observe $\bW = \bW_0 + \beta \bu\bu^{\top}$ for $\bW_0 \sim \GOE(n)$ and some fixed $\beta > 0$, held constant as $n \to \infty$.
\end{itemize}

Two natural statistical questions arise: (1) \emph{detection} or \emph{testing}, where we observe $\bW$ drawn from $\PP$ or $\QQ$ and must decide which distribution $\bW$ was drawn from, and (2) \emph{recovery} or \emph{estimation}, where we observe $\bW \sim \PP$ and seek to produce a good estimate of $\bu$.
In either case, the associated optimization problem of computing the largest eigenvalue is natural to consider:
\begin{equation}
    \label{eq:lambda-max}
    \lambda_{\max}(\bW) \colonequals \max_{\|\bx\|_2 = 1} \bx^{\top} \bW \bx.
\end{equation}
For recovery, computing the maximizer $\bx^{\star}$ (the top eigenvector of $\bW$) performs maximum likelihood estimation of $\bu$.
For detection, computing and thresholding $\lambda_{\max}(\bW)$ itself is a natural and often effective strategy.
The optimal value and optimizer of $\lambda_{\max}(\bW)$ can be approximated (to arbitrary accuracy) in time $\poly(n)$, so these correspond to efficient algorithms for recovery and detection, respectively.
Moreover, these algorithms are essentially \emph{optimal}: almost whenever\footnote{With the exception of the critical case $\beta = 1$, where thresholding $\lambda_{\max}(\bW)$ does not distinguish $\PP$ from $\QQ$, but it is possible to do so by considering more sophisticated statistics \cite{OJ-2020-TestingSpikedModels}.} it is possible to distinguish $\PP$ from $\QQ$ with high probability,\footnote{We say that a sequence of events $A_n$ occurs \emph{with high probability} under a sequence of probability measures $\PP_n$ if $\lim_{n \to \infty} \PP_n[A_n] = 1$.} thresholding $\lambda_{\max}(\bW)$ achieves this; whenever it is possible to estimate $\bu$ non-trivially then the optimizer $\bx^{\star}$ achieves this.
\begin{proposition}[\cite{FP-2007-LargestEigenvalueWigner,CDMF-2009-DeformedWigner,MRZ-2015-LimitationsSpectral,BMVVX-2018-InfoTheoretic}]
    Let $\bW \sim \PP$ with parameter $\beta \ge 0$ (note that taking $\beta = 0$ gives $\PP = \QQ$), and let $\bu$ be the spike vector.
    Let $\bx^{\star}(\bW)$ be the optimizer of $\lambda_{\max}(\bW)$, the top eigenvector of $\bW$ scaled to have unit norm.
    Then, we have almost surely
    \begin{align}
      \lim_{n \to \infty} \lambda_{\max}(\bW) &= \left\{\begin{array}{ll} 2 & \text{if } 0 \leq \beta \leq 1, \\ \beta + \beta^{-1} > 2 \, \, \, \, \, \, & \text{if } \beta > 1, \end{array} \right. \\
      \lim_{n \to \infty} |\la \bx^{\star}(\bW), \bu \ra| &= \left\{\begin{array}{ll} 0 & \text{if } 0 \leq \beta \leq 1, \\ \sqrt{1 - \beta^{-2}} > 0 & \text{if } \beta > 1. \end{array}\right.
    \end{align}
    Moreover, if $0 \le \beta < 1$ then there is no function of $\bW$ that with high probability selects correctly whether $\bW$ is drawn from $\PP$ or $\QQ$, and if $0 \le \beta \leq 1$ then there there is no unit vector-valued function of $\bW$ that has inner product with $\bx$ asymptotically bounded away from zero with high probability when $\bW \sim \PP$.
\end{proposition}
\noindent
More specifically, the behavior of $\lambda_{\max}(\bW)$ is established by \cite{FP-2007-LargestEigenvalueWigner}, while the behavior of $|\langle \bx^{\star}(\bW), \bu \rangle|$ is determined by \cite{CDMF-2009-DeformedWigner} (both building on the seminal results of \cite{BBAP-2005-LargestEigenvalueSampleCovariance}).
The impossibility of ``detection'' or of selecting whether $\bW$ is drawn from $\PP$ or $\QQ$ with high probability is shown by \cite{MRZ-2015-LimitationsSpectral} by establishing \emph{contiguity} of these two sequences of probability measures.
Finally, \cite{BMVVX-2018-InfoTheoretic} show that this contiguity implies the impossibility of estimating $\bu$ with positive correlation.

More refined models follow from choosing more structured distributions of $\bu$.
This corresponds to extracting principal components under some prior knowledge of their structure.
One natural example was introduced by \cite{MR-2015-NonNegative}, where $\bu$ is chosen uniformly from the positive orthant of $\mathbb{S}^{n - 1}$ instead of the entire sphere, which yields the problem of \emph{non-negative PCA}.
Here, the null model $\QQ$ remains as above, while $\PP$ is replaced with $\PP^+$ defined as follows:
\begin{itemize}
    \item Under $\PP^+$, first draw $\bv \sim \Unif(\mathbb{S}^{n - 1})$, let $\bu$ have entries $u_i = |v_i|$, and then observe $\bW = \bW_0 + \beta \bu\bu^{\top}$ for $\bW_0 \sim \GOE(n)$ and some fixed $\beta > 0$.
\end{itemize}
Following the case of classical PCA, we might hope to attack detection and recovery by solving the optimization problem
\begin{equation}
    \label{eq:lambda-plus}
    \lambda^{+}(\bW) \colonequals \max_{\substack{\|\bx\|_2 = 1 \\ \bx \geq 0}} \bx^\top \bW \bx,
\end{equation}
where $\bx \geq 0$ means $x_i \geq 0$ for each $i \in [n]$.
Here, however, a crucial difference between classical PCA and non-negative PCA arises: unlike $\lambda_{\max}(\bW)$ and the associated optimizer, it is $\mathsf{NP}$-hard to compute $\lambda^+(\bW)$ for general $\bW$ \cite{KP-2002-CopositiveProgramming}.
Therefore, non-negative PCA poses a more substantial algorithmic challenge.

Nonetheless, using an approximate message-passing (AMP) algorithm \cite{MR-2015-NonNegative} showed that it is possible to solve this problem essentially to optimality for random inputs from $\QQ$ or $\PP^+$.
We focus now just on the ``null'' case $\bW \sim \QQ = \GOE(n)$.
\begin{proposition}[Part of Theorem 2 and Proposition 4.5 of \cite{MR-2015-NonNegative}]
    \label{prop:true-value}
    Almost surely for $\bW \sim \GOE(n)$,
    \begin{equation}
        \lim_{n \to \infty} \lambda^+(\bW) = \sqrt{2}.
    \end{equation}
    Moreover, for any $\varepsilon > 0$, there exists an algorithm that runs in time $\poly(n)$ (with runtime also depending on $\varepsilon$) and computes $\bx \in \mathbb{S}^{n - 1}$ with $\bx \geq 0$ that has $\bx^{\top}\bW \bx \geq \sqrt{2} - \varepsilon$ with high probability.
\end{proposition}
\noindent
This same algorithm is also effective for detection between $\QQ$ and $\PP^+$ and recovery under $\PP^+$.

\begin{remark}
    One may check that various simpler algorithms do not produce a solution of the same quality.
    For example, if $\bv$ is the top eigenvector of $\bW$, then one simple algorithm is to take $\bv^+$ having entries $v^+_i = \max(0, v_i)$ and return $\bx = \bv^+ / \|\bv^+\|$.
    However, computing heuristically, we have
    \begin{equation}
        \bx^{\top}\bW \bx \approx \frac{1}{\|\bv^+\|^2} \cdot 2 \langle \bv, \bv^+ \rangle^2 \approx 2 \cdot 2 \cdot \left(\frac{1}{2}\right)^2 = 1,
    \end{equation}
    whereby this choice of $\bx$ is inferior to that produced by AMP.
\end{remark}

In this paper, we study an alternative to AMP, also suggested in \cite{MR-2015-NonNegative}, where we substitute for the intractable optimization problem $\lambda^+(\bW)$ the following tractable convex relaxation, a natural \emph{semidefinite program (SDP)}:
\begin{equation}\label{eq:SDP}
    \SDP(\bW) \colonequals \max_{\substack{\bX \succeq 0 \\ \bX \geq 0 \\ \Tr(\bX) = 1}} \la \bX, \bW \ra \geq \lambda^+(\bW).
\end{equation}
In \cite{MR-2015-NonNegative}, numerical experiments are presented that suggest that this SDP is effective in recovering $\bx$ under $\PP^+$, the restriction of the top eigenvector of the optimizer $\bX^{\star}$ to only positive entries giving a comparable estimate of $\bu$ to the AMP algorithm (see their Section 5.3).
In Section~\ref{sec:finite-size}, we present analogous experiments for $\bW \sim \GOE(n)$, and note that in this case the SDP is often \emph{tight}, the optimizer $\bX^{\star}$ being rank one within numerical tolerances.
While the SDP is much slower than AMP, its apparent efficacy is nevertheless tantalizing, suggesting that the algorithmic tractability of non-negative PCA might be unified with other situations where SDP relaxations of maximum likelihood estimation are tight \cite{BKS-2014-TightnessMLESDP}. Furthermore, the SDP offers some advantages over AMP: it algorithmically \emph{proves} (or \emph{certifies}) an upper bound on the value of $\lambda^+(\bW)$, and it may also exhibit robustness properties that SDPs have been shown to enjoy in other settings~\cite{FK-2001-SemirandomGraphProblems,MPW-2016-RobustReconstruction,RTJM-2016-SDPRobustness}.
The problem of determining the asymptotic behavior of the SDP as $n \to \infty$ under $\bW \sim \GOE(n)$ was also posed explicitly as Open Problem~9.5 in the lecture notes \cite{Bandeira-2015-TenLectures} of one of the authors.

We therefore take up the following two questions concerning this SDP and related algorithmic approaches when $\bW \sim \GOE(n)$.
The first concerns this specific semidefinite program:
\begin{enumerate}
    \item[1.] When $\bW \sim \GOE(n)$, does $\SDP(\bW) \to \sqrt{2}$ in probability as $n \to \infty$?
\end{enumerate}
As we will see, the answer in the limit $n \to \infty$, in surprising contrast to the experiments of~\cite{MR-2015-NonNegative} for small $n$, is \emph{no}.
In fact, we instead have $\SDP(\bW) \to \lambda_{\max}(\bW) \approx 2$ as $n \to \infty$.
We then ask whether any remotely efficient algorithm that \emph{certifies} upper bounds on $\lambda^+(\bW)$ can improve upon this.
\begin{enumerate}
    \item[2.] Does there exist an algorithm that runs in time $\exp(O(n^{1 - \eta}))$ for some fixed $\eta > 0$ and computes $c: \RR^{n \times n}_{\sym} \to \RR$ with the following two properties?
    \begin{itemize}
        \item For \emph{all} $\bW \in \RR^{n \times n}_{\sym}$, we have $c(\bW) \geq \lambda^+(\bW)$.
        \item When $\bW \sim \GOE(n)$, we have $c(\bW) \leq 2 - \varepsilon$ with high probability for some fixed $\varepsilon > 0$.
    \end{itemize}
\end{enumerate}
We provide rigorous evidence, based on the \emph{low-degree polynomial method}, that even the answer to this much broader question again is \emph{no}.

\subsection{Organization}

The remainder of the paper is organized as follows.
In Section~\ref{sec:sdp}, we state and prove a lower bound on $\SDP(\bW)$ when $\bW \sim \GOE(n)$.
In Section~\ref{sec:low-deg}, we state and prove a \emph{reduction} from a certain hypothesis testing problem to the problem of certifying bounds on $\lambda^+(\bW)$, and review evidence from prior work that this hypothesis testing problem is computationally hard.
Finally, in Section~\ref{sec:finite-size}, we present the results of larger numerical experiments that capture the departure from the ``tight regime'' where the optimizer of $\SDP$ has rank one, a striking example of the difference between theoretical asymptotics and computations tractable at ``laptop scale'' for semidefinite programming.

\section{Lower Bound on Semidefinite Programming}
\label{sec:sdp}

In this section we will prove the following result, which gives the asymptotic value of $\SDP(\bW)$.
\begin{theorem}
    \label{thm:sdp}
    For any $\varepsilon > 0$, $\lim_{n \to \infty} \PP[2 - \epsilon \leq \SDP(\bW) \leq 2 + \epsilon] = 1$ where $\bW \sim \GOE(n)$.
\end{theorem}

The main technical tool required will be the following concentration inequality for the entries of a random projection matrix.
\begin{proposition}
    \label{prop:proj-entries}
    Let $\delta \in (0, 1)$.
    Let $\bP \in \RR^{n \times n}$ be the orthogonal projection matrix to a Haar-distributed subspace of $\RR^n$ having dimension $r \colonequals \delta n$.
    Then, for any $K > 0$, there exist constants $C_{\delta, K}, C^{\prime}_{\delta, K} > 0$ such
    \begin{equation}
        \PP\left[\max_{i, j \in [n]} \left\{\begin{array}{lll} \left|P_{ii} - \delta\right| & \text{if} & i = j \\ \left|P_{ij}\right| & \text{if} & i \neq j\end{array}\right\} \leq C_{\delta,K}\sqrt{\frac{\log n}{n}} \right] \geq 1 - \frac{C^{\prime}_{\delta,K}}{n^{K}}.
    \end{equation}
\end{proposition}
\noindent
See, e.g., \cite{KB-2019-Degree4SK} for a careful proof.

\begin{proof}[Proof of Theorem~\ref{thm:sdp}]
    For the upper bound, note that $\langle \bX, \bW \rangle \leq \lambda_{\max}(\bW)$ for any $\bX$ feasible for the SDP.
    The bound then follows from standard bounds on the spectrum of $\bW$ \cite{AGZ-2010-RandomMatrices}.

    For the lower bound, fix $\alpha, \delta \in (0, 1)$.
    Let $r = \delta n$, assuming for the sake of simplicity that this is an integer.
    Let $\bP$ be the orthogonal projector to the span of the $r$ eigenvectors of $\bW$ having the largest eigenvalues.
    Then, define
    \begin{equation}
        \label{eq:X-witness-def}
        \bX = \bX^{(\alpha, \delta)} \colonequals (1 - \alpha) \frac{1}{r}\bP + \alpha \frac{1}{n}\one_n\one_n^\top,
    \end{equation}
    where $\one_n \in \RR^n$ is the vector with all entries equal to 1.
    Denote by $F$ the event that $\bX$ is feasible for the SDP, i.e., the event that $\Tr(\bX) = 1$, $\bX \succeq \bm 0$, and $\bX \geq \bm 0$.

    We first show that, for any fixed $\alpha, \delta \in (0, 1)$, $\PP[F] \to 1$.
    Since $\Tr(\bP) = r$ we have $\Tr(\bX) = 1$, and $\bX \succeq \bm 0$ since $\bX$ is a convex combination of two positive semidefinite matrices.
    In particular, $X_{ii} \geq 0$ for all $i \in [n]$.
    For the off-diagonal entries, we observe that the range of $\bP$ is a Haar-distributed $r$-dimensional subspace of $\RR^n$.
    Thus by Proposition~\ref{prop:proj-entries}, with high probability, for all $i, j \in [n]$ with $i \neq j$,
    \begin{equation}
    X_{ij} \geq \frac{\alpha}{n} - \frac{1 - \alpha}{r}|P_{ij}| \geq \frac{\alpha}{n} - C_{\delta}(1 - \alpha)\frac{1}{n}\sqrt{\frac{\log n}{n}} \end{equation}
    for some constant $C_{\delta} > 0$ depending only on $\delta$.
    In particular, for $\alpha, \delta$ fixed as $n \to \infty$, this is non-negative for all sufficiently large $n$, thus $X_{ij} \geq 0$ for all $i, j \in [n]$ with high probability.
    Combining these observations, we find that $F$ occurs with high probability.

    On the event $F$, we have
    \begin{equation}
        \SDP(\bW) \geq \la \bW, \bX \ra = (1 - \alpha) \cdot \frac{1}{r}\sum_{i = 1}^r \lambda_i(\bW) + \alpha \frac{1}{n}\one_n^\top \bW \one_n.
    \end{equation}
    The second term is distributed as $\sN(0, 2\alpha^2 / n)$, and the first term is with high probability bounded below by $(1 - \alpha)(2 - f(\delta))$ for some $f(\delta)$ with $\lim_{\delta \to 0}f(\delta) = 0$, by the convergence of the law of the empirical spectrum of $\bW$ to the semicircle distribution \cite{AGZ-2010-RandomMatrices}.
    In particular, for any $\varepsilon > 0$, we may choose $\alpha, \delta \in (0, 1)$ sufficiently small that the above argument shows $\SDP(\bW) \geq 2 - \varepsilon$ with high probability.
\end{proof}
\noindent
We note that the general proof technique of ``nudging'' an initial construction that is not feasible for a convex program towards a deterministic feasible point has been used before for relaxations of the cut polytope where the latter point is the identity matrix  \cite{AU-2003-LPMaxCut,KB-2019-Degree4SK,MRX-2019-SOS4}.
Our proof adapts this to our different SDP constraints by using the all-ones matrix for this purpose instead.

\begin{question}[Sum-of-squares relaxations]
    It would be interesting to extend this result to lower bounds on higher-degree sum-of-squares relaxations of $\lambda^+(\bW)$; based on our results in Section~\ref{sec:low-deg}, it is natural to conjecture that no relaxation of constant degree certifies a bound strictly smaller than 2 on $\lambda^+(\bW)$ as $n \to \infty$ (since this would refute Conjecture~\ref{conj:cert}).
    We remark that, in working with inequality constraints, there are a number of reasonable ways to formulate a sum-of-squares relaxation of given degree; see, e.g., \cite{Laurent-2009-SOS,OZ-2013-ApproximabilityProofComplexity} for some discussion of these details.
    To the best of our knowledge, lower bounds for sum-of-squares relaxations with inequality constraints have not been studied for high-dimensional random problems, so this problem would be a convenient testing ground to see whether these nuances play an important technical role.
\end{question}

\section{Evidence for General Hardness of Certification}
\label{sec:low-deg}

We first formalize the notion of a certification algorithm.
\begin{definition}[Certification algorithm]
\label{def:cert}
    Suppose an algorithm takes as input $\bW \in \RR^{n \times n}_{\sym}$ and outputs a number $c(\bW) \in \RR$ such that $c(\bW) \geq \lambda^+(\bW)$ for all $\bW \in \RR^{n \times n}_{\sym}$.
    If when $\bW \sim \GOE(n)$ then $c(\bW) \leq K$ with high probability as $n \to \infty$, then we say that this algorithm \emph{certifies} the bound $\lambda^+(\bW) \le K$.
\end{definition}
\noindent
The key property of a certification algorithm is that it must give a valid upper bound on $\lambda^+(\bW)$ no matter what input matrix $\bW$ is supplied; in particular, it must even do so for $\bW$ that are atypical under the distribution $\GOE(n)$. However, this upper bound only needs to be a ``good'' bound for typical $\bW \sim \GOE(n)$. One notable class of certification algorithms are convex relaxations such as the SDP~\eqref{eq:SDP}. Note that $\lambda_{\max}(\bW)$ certifies the bound $\lambda^+(\bW) \le 2 + o(1)$. The goal of this section is to provide formal evidence for the following conjecture, which states that this simple certificate cannot be improved except by a fully exponential-time ``brute force'' search.

\begin{conjecture}\label{conj:cert}
For any fixed $\varepsilon > 0$ and $\eta > 0$, there is no algorithm of runtime $\exp(O(n^{1-\eta}))$ that certifies the bound $\lambda^+(\bW) \le 2 - \varepsilon$ (in the sense of Definition~\ref{def:cert}).
\end{conjecture}

\noindent We will argue that the certification problem is hard by reduction from a particular hypothesis testing problem, which we define next.
\begin{definition}[Centered Bernoulli distribution]
For a constant $\rho \in (0,1)$, let $\mathcal{X}_\rho$ be the distribution over $\RR^n$ where $\bu \sim \sX_{\rho}$ is drawn by drawing each coordinate $u_i$ independently as
\begin{equation}
    u_i = \left\{\begin{array}{ll} \sqrt{\frac{1-\rho}{\rho n}} & \text{with probability } \rho, \\ -\sqrt{\frac{\rho}{(1-\rho)n}} & \text{with probability } 1-\rho. \end{array}\right.
\end{equation}
\end{definition}
\noindent
This is scaled so that $\EE[u_i] = 0$ and $\|\bu\| \to 1$ in probability.
\begin{definition}
    Given constants $\gamma > 0$ and $\beta > -1$, the \emph{spiked Wishart model} with spike prior $\sX_{\rho}$ consists of the following pair of probability distributions.
    Let $N = N(n) \in \NN$ such that $n / N \to \gamma$ as $n \to \infty$.
    \begin{itemize}
    \item Under $\QQ$, draw $\by_1, \dots, \by_N \sim \sN(\bm 0, \bm I_n)$ independently.
    \item Under $\PP$, first draw $\bu \sim \sX_{\rho}$.
        If $\beta \|\bu\|^2 \leq -1$, draw $\by_1 = \cdots = \by_N = 0$.
        Otherwise, draw $\by_1, \dots, \by_N \sim \sN(\bm 0, \bm I_n + \beta \bu\bu^{\top})$ independently (noting that the covariance matrix is positive definite).
    \end{itemize}
    If $\beta < 0$, we call such a model a \emph{negatively-spiked Wishart model}.
\end{definition}

\noindent We will consider the \emph{strong detection} problem where the goal is to give a test $f: \RR^{n \times N} \to \{\ttp, \ttq\}$ that takes input $\by = (\by_1,\ldots,\by_N)$ and distinguishes between $\PP$ and $\QQ$ with error probability $o(1)$, i.e.,
\begin{equation}\label{eq:strong-det}
\lim_{n \to \infty} \PP[f(\by) = \ttp] = \lim_{n \to \infty} \QQ[f(\by) = \ttq] = 1.
\end{equation}
When $\beta^2 > \gamma$ (the ``BBP transition''), it is well-known that strong detection is possible in polynomial time via the maximum (if $\beta > 0$) or minimum (if $\beta < 0$) eigenvalue of the sample covariance matrix~\cite{BBAP-2005-LargestEigenvalueSampleCovariance,BS-2006-EigenvaluesSampleCovariance}. While strong detection is sometimes (depending on $\rho$) possible when $\beta^2 < \gamma$ via brute force search, this is conjectured to be impossible in subexponential time.

\begin{conjecture}[\cite{BKW-2019-ConstrainedPCA} Conjecture 3.1, Corollary 3.3]
\label{conj:wishart}
For any constants $\gamma > 0$, $\beta > -1$, $\rho \in (0,1)$, $\eta > 0$ such that $\beta^2 < \gamma$, there is no algorithm of runtime $\exp(O(n^{1-\eta}))$ that achieves strong detection in the spiked Wishart model with parameters $\gamma,\beta$ and spike prior $\mathcal{X}_\rho$.
\end{conjecture}

\noindent This conjecture is justified in~\cite{BKW-2019-ConstrainedPCA} by formal evidence based on the \emph{low-degree polynomial method}, a framework based on~\cite{pcal,HS-bayesian,sos-hidden,sam-thesis} that has been successful in predicting and explaining computational hardness in a wide variety of tasks in high-dimensional statistics; see~\cite{KWB-2019-NotesLowDegree} for a survey. More precisely, it is shown (Theorem~3.2 of~\cite{BKW-2019-ConstrainedPCA}) that when $\beta^2 < \gamma$, no multivariate polynomial $f: \RR^{n \times N} \to \RR$ of degree $D = o(n/\log n)$ can distinguish $\PP$ and $\QQ$ in the sense of $\EE_\PP[f(\by)] \to \infty$ while $\EE_\QQ[f(\by)^2] = 1$. (This is true not only for the centered Bernoulli prior but more generally for any spike prior where the $u_i$ are distributed i.i.d.\ as $\frac{1}{\sqrt n} \pi$ for a fixed distribution $\pi$ on $\RR$ that is subgaussian with $\EE[\pi] = 0$ and $\EE[\pi^2] = 1$.) Degree-$D$ polynomial tests of the above form are believed to be as powerful as any $\exp(\tilde\Omega(D))$-time algorithm (where $\tilde\Omega$ hides factors of $\log n$) for a broad class of high-dimensional testing problems; see~\cite{sam-thesis,KWB-2019-NotesLowDegree,subexp-sparse}.

We are now prepared to state the main result of this section, which shows that if it is possible to certify a bound on $\lambda^+(\bW)$ below 2, then it is possible to produce a test between $\PP$ and $\QQ$ in a particular negatively-spiked Wishart model whose parameters lie in the ``hard'' regime $\beta^2 < \gamma$. An immediate consequence is that Conjecture~\ref{conj:wishart} implies Conjecture~\ref{conj:cert}.

\begin{theorem}[Reduction from detection to certification]
    \label{thm:reduction}
    Suppose there exists a constant $\varepsilon > 0$ and a $t(n)$-time certification algorithm $c: \RR^{n \times n}_{\sym} \to \RR$ for $\lambda^+$ such that, with high probability as $n \to \infty$, $c(\bW) \leq 2 - \varepsilon$ when $\bW \sim \GOE(n)$.
    Then there exist constants $\gamma > 1$, $\beta \in (-1,0)$, and $\rho \in (0,1)$ (depending on $\varepsilon$) such that there is a $(t(n)+\mathsf{poly}(n))$-time algorithm computing $f: \RR^{n \times N} \to \{\ttp, \ttq\}$ that achieves \emph{strong detection} (in the sense of~\eqref{eq:strong-det}) in the negatively-spiked Wishart model with parameters $\gamma,\beta$ and spike prior $\mathcal{X}_\rho$.
\end{theorem}

\noindent The proof is similar to that of Theorem~3.8 in our prior work~\cite{BKW-2019-ConstrainedPCA} (which gives the analogous result when the constraint set is $\{\pm 1/\sqrt{n}\}^n$ instead of the positive orthant) with one key difference. As in~\cite{BKW-2019-ConstrainedPCA}, the idea of the reduction is to create a GOE matrix whose top eigenspace has been ``rotated'' to align with the orthogonal complement of the span of the given Wishart samples. If the samples come from $\PP$ (with $\beta$ slightly greater than $-1$ and $\gamma$ slightly greater than $1$) then the Wishart samples are nearly orthogonal to the planted vector $\bu$, so this has the effect of planting $\bu$ in the top eigenspace of the matrix. We would like to plant a non-negative vector in the top eigenspace so that any certifier is forced to output a bound larger than $2-\varepsilon$. However, we cannot take $\bu$ to be non-negative because it is important for Wishart hardness (Conjecture~\ref{conj:wishart}) that $\bu$ have mean zero. The key idea is to instead choose $\bu$ to be a mean-zero vector that is highly correlated with a certain non-negative vector $\widehat{\bz}$; this is the purpose of the centered Bernoulli prior $\mathcal{X}_\rho$.

\begin{proof}[Proof of Theorem~\ref{thm:reduction}]
    Suppose a certification algorithm as stated exists.
    We will use this to design a test $f$ achieving strong detection in the negatively-spiked Wishart model.

    Call $\by_1, \dots, \by_N$ the samples from the Wishart model.
    Draw $\widetilde{\bW} \sim \GOE(n)$ and let $\lambda_1 \le \cdots \le \lambda_n$ be its eigenvalues.
    Let $\bv_1,\dots,\bv_N$ be a uniformly random orthonormal basis for $V \colonequals \mathsf{span}(\{\by_1,\ldots,\by_N\})$ and let $\bv_{N+1},\dots,\bv_n$ be a uniformly random orthonormal basis for the orthogonal complement $V^\perp$.
    Let $\bW \colonequals \sum_{i=1}^n \lambda_i \bv_i \bv_i^\top$.

    Then, using the certification algorithm computing $c(\bW)$ as a subroutine, we compute the test $f$ as
    \begin{equation}
        f(\bW) \colonequals \left\{\begin{array}{ll} \ttq & \text{if } c(\bW) \leq 2 - \varepsilon, \\ \ttp & \text{otherwise}. \end{array} \right.
    \end{equation}
    When $(\by_1, \dots, \by_N) \sim \QQ$, then $\bW$ has the law $\GOE(n)$, so $c(\bW) \leq 2 - \varepsilon$ with high probability, and thus $f(\bW) = \ttq$ with high probability.
    Thus to complete the proof it suffices to show that, when $(\by_1, \dots, \by_N) \sim \PP$, then $c(\bW) > 2 - \varepsilon$ with high probability, whereby we will have $f(\bW) = \ttp$ with high probability.

    To this end, suppose $(\by_1, \dots, \by_N) \sim \PP$ with $\bu$ the spike vector.
    Let $\bz \ge 0$ be the vector
    \begin{equation}
        z_i = \left\{\begin{array}{ll} 1 / \sqrt{\rho n} & \text{if } u_i > 0 \\ 0 & \text{otherwise}, \end{array}\right.
    \end{equation}
    and note that $\|\bz\| \to 1$ and $\langle \bz,\bu \rangle \to \sqrt{1-\rho}$ in probability.
    Let $\widehat{\bz} \colonequals \bz / \|\bz\|$.
    We then have
    \begin{align}
      c(\bW)
      &\ge \lambda^+(\bW) \nonumber \\
      &\geq \widehat{\bz}^\top \bW \widehat{\bz} \nonumber \\
      &= \sum_{i=1}^n \lambda_i \langle \widehat{\bz}, \bv_i \rangle^2 \nonumber \\
      &\ge \lambda_1 \sum_{i=1}^N \langle \widehat{\bz}, \bv_i \rangle^2 + \lambda_{N+1} \sum_{i=N+1}^n \langle \widehat{\bz}, \bv_i \rangle^2 \nonumber \\
      &\ge \lambda_1 \sum_{i=1}^N \langle \widehat{\bz}, \bv_i \rangle^2 + \lambda_{N+1} \left(1 - \sum_{i=1}^N \langle \widehat{\bz}, \bv_i \rangle^2 \right) \nonumber \\
      \intertext{and, since $\{\bv_1,\ldots,\bv_n\}$ is an orthonormal basis and $\|\widehat{\bz}\| = 1$,}
      &= \lambda_{N+1} - (\lambda_{N+1} - \lambda_1) \sum_{i=1}^N \langle \widehat{\bz},\bv_i \rangle^2.
    \end{align}
    Recalling that $\{\bv_i\}_{i=1}^N$ is an orthonormal basis for $\mathsf{span}(\{\by_1,\ldots,\by_N\})$, we have $\sum_{i=1}^N \bv_i \bv_i^\top \preceq \frac{1}{\mu} \bY$ where $\bY = \frac{1}{N} \sum_{i=1}^N \by_i \by_i^\top$ and $\mu$ is the smallest nonzero eigenvalue of $\bY$.
    We therefore have
    \begin{equation}
        \sum_{i=1}^N \langle \widehat{\bz},\bv_i \rangle^2 \le \frac{1}{\mu N} \sum_{i=1}^N \langle \widehat{\bz},\by_i \rangle^2.
    \end{equation}
    Also, viewing $\bY$ as a sample covariance matrix under a spiked matrix model, we have $\mu \to (\sqrt{\gamma}-1)^2 > 0$ in probability by Theorem~1.2 of~\cite{BS-2006-EigenvaluesSampleCovariance}.

    For fixed $\widehat{\bz}$, since $\by_i \sim \mathcal{N}(0,\bm I_n + \beta \bu \bu^\top)$ we have that $\langle \widehat{\bz},\by_i \rangle$ is Gaussian with mean zero and variance
    \begin{equation}
        \EE\langle \widehat{\bz},\by_i \rangle^2 = \EE[\widehat{\bz}^\top \by_i \by_i^\top \widehat{\bz}] = \widehat{\bz}^\top \EE[\by_i \by_i^\top] \widehat{\bz} = \widehat{\bz}^\top (\bm I_n + \beta \bu \bu^\top) \widehat{\bz} \to 1 + \beta(1-\rho)
    \end{equation}
    in probability, where we have used $\langle \bz,\bu \rangle \to \sqrt{1-\rho}$ and $\|\bz\| \to 1$.
    Therefore $\frac{1}{N} \sum_{i=1}^N \langle \widehat{\bz},\by_i \rangle^2 \to 1 + \beta(1-\rho)$ in probability as well.

    By the convergence of $\{\lambda_i\}$ to the semicircle law on $[-2,2]$ \cite{AGZ-2010-RandomMatrices}, we have $\lambda_{N+1} - \lambda_1 \le 5$ with high probability.
    Also, by choosing $\gamma > 1$ sufficiently close to $1$, we can ensure $\lambda_{N+1} \ge 2 - \varepsilon/2$ with high probability.

    Putting it all together, the value of $c(\bW)$ under $(\by_1, \dots, \by_N) \sim \PP$ with high probability satisfies
    \[ c(\bW) \ge \lambda_{N+1} - (\lambda_{N+1} - \lambda_1) \sum_{i=1}^N \langle \widehat{\bz},\bv_i \rangle^2 \ge 2-\varepsilon/2 - 5(\sqrt{\gamma}-1)^{-2}\big(1 + \beta(1-\rho)\big), \]
    which exceeds $2-\varepsilon$ provided we choose $\beta > -1$ close enough to $-1$ and $\rho > 0$ small enough.
\end{proof}

\section{Deceptive Finite-Size Effects}
\label{sec:finite-size}

Experiments for small $n$ computing $\SDP(\bW)$ suggest, contrary to our results, that $\SDP(\bW)$ is in fact very effective in bounding $\lambda^+(\bW)$ for $\bW \sim \GOE(n)$.
Not only do we observe for small $n$ the value $\SDP(\bW) \approx \sqrt{2}$, but we also find that $\SDP(\bW)$ appears to be \emph{tight}, the primal optimizer $\bX^{\star}$ often having rank one to numerical tolerances.
This is analogous to the efficacy of $\SDP(\bW)$ for recovering the spike for $\bW$ under the spiked non-negative PCA model discussed in \cite{MR-2015-NonNegative}.
We describe these results here, as well as further experiments suggesting what size of $n$ is required for these finite size effects to give way to the correct asymptotics.

The first experiment we consider solves $\SDP(\bW)$ for many random choices of $\bW$, obtains the optimizer $\bX^{\star}$, and considers the \emph{numerical rank} of $\bX^{\star}$.
We plot the results of 50 trials of this experiment with $n = 150$ in Figure~\ref{fig:primal-hist}, and observe that most trials have the second-largest eigenvalue of $\bX^{\star}$ of order at most $10^{-4}$ compared to the trace of 1, whereby $\bX^{\star}$ is nearly rank-one and the SDP is nearly tight.

The next experiment solves the following different SDP, which is dual to $\SDP(\bW)$, and which by a standard strong duality argument has the same value as $\SDP(\bW)$:
\begin{equation}
    \SDP^*(\bW) \colonequals \min_{\bY \geq 0} \lambda_{\max}(\bW + \bY) = \SDP(\bW).
\end{equation}
Having $\SDP(\bW) \leq 2 - \varepsilon$ therefore has the elegant interpretation of it being possible to ``compress'' the spectrum of $\bW \sim \GOE(n)$ below 2 by only increasing each entry.
We plot the results of 50 trials of this experiment with $n = 150$ in Figure~\ref{fig:dual-hist}.
(These semidefinite programs are solved using version 9.2 of the \texttt{Mosek} solver on a laptop computer with 32GB RAM and an Intel i7-1065G7 processor; the average time to solve an instance is 18.1 minutes.)
From these results, this compression indeed appears possible; moreover, the compressed spectrum appears to have an interesting ``wall shape'' not unlike that of the GOE conditioned on its largest eigenvalue being small; see, e.g., Figure~4 of \cite{MS-2014-LargeDeviationsTopEigenvalue}.

The final experiment seeks to identify how large we would need to make $n$ in order to observe that the above are all illusory finite-size effects.
To do this, we consider our primal witness from the proof of Theorem~\ref{thm:sdp},
\begin{equation}
    \label{eq:Xad-def}
    \bX^{(\alpha, \delta)} = \frac{1 - \alpha}{\delta n}\bm P + \frac{\alpha}{n} \one_n\one_n^{\top}
\end{equation}
for $\bP$ the projection matrix to the top $\delta n$ eigenvectors of $\bW \sim \GOE(n)$ and $\alpha$ large enough that $\bX^{(\alpha, \delta)} \geq 0$ entrywise.
In Figure~\ref{fig:lower-bound}, we fix $\delta = 1/25 = 0.04$, and plot both the smallest $\alpha$ making $\bX^{(\alpha, \delta)}$ feasible for the SDP and the corresponding lower bound $\langle \bX^{(\alpha, \delta)}, \bW \rangle$ on the SDP.
We see that the smallest $\alpha$ only decays to zero very slowly, as $\widetilde{O}(n^{-1/2})$ per our argument.
Accordingly, $\langle \bX^{(\alpha, \delta)}, \bW \rangle$ also only very slowly approaches its limiting value.
Moreover, even to have $\langle \bX^{(\alpha, \delta)}, \bW \rangle > \sqrt{2}$ requires $n \sim 10^4$, suggesting that this is roughly the size of $n$ required to observe that $\SDP(\bW)$ in fact is not typically tight (in stark contrast to $n \sim 10^2$ that is tractable to solve on commodity hardware).

Taken together, these experiments present a striking caution against extrapolating asymptotic behavior for an SDP from experimental results tractably computable in reasonable time in practice.
In our case, the correct asymptotic behavior ``kicks in'' only for problems two orders of magnitude larger than the largest tractable with off-the-shelf software on a personal computer.

\section*{Acknowledgements}
\addcontentsline{toc}{section}{Acknowledgements}

We thank Yash Deshpande, Sam Hopkins, and Andrea Montanari for helpful discussions.

\newpage

\begin{figure}[!h]
    \centering
    \includegraphics[scale=0.85]{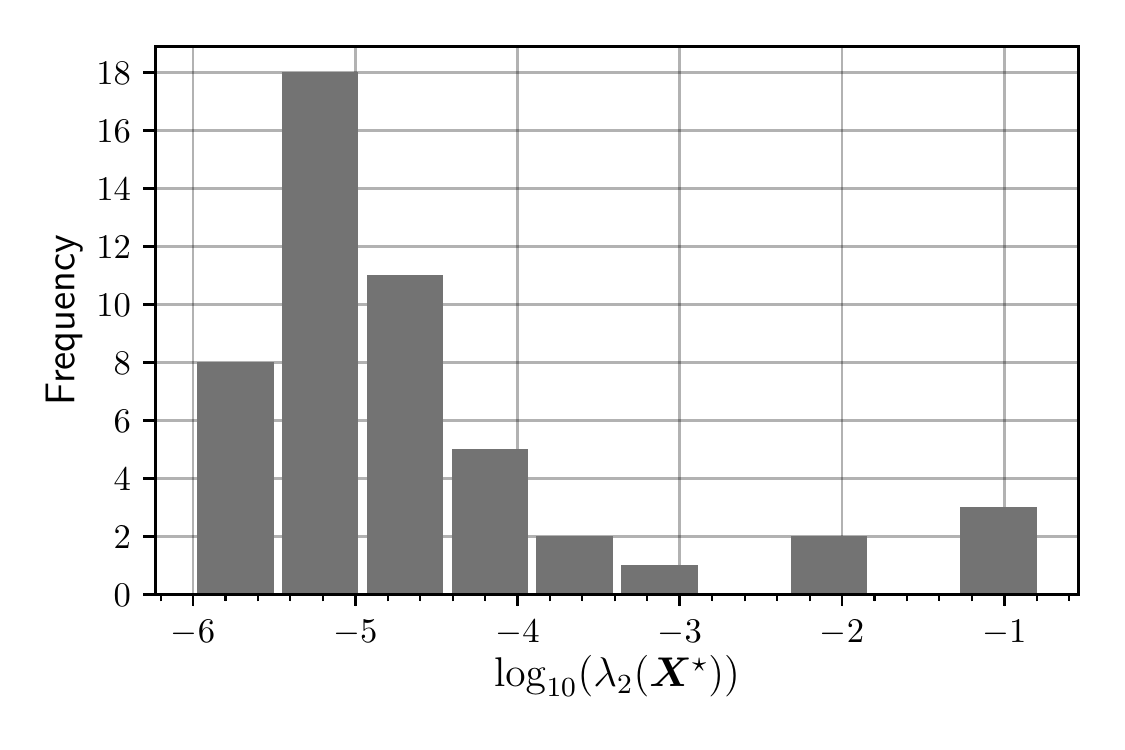}
    \vspace{-1.5em}
    \caption{\textbf{SDP primal spectrum for small $n$.} We plot a histogram of the second-largest eigenvalue of $\bX^{\star}$ the optimizer of $\SDP(\bW)$, over 50 trials with $n = 150$.
    }
    \label{fig:primal-hist}
\end{figure}

\begin{figure}[!h]
    \centering
    \includegraphics[scale=0.85]{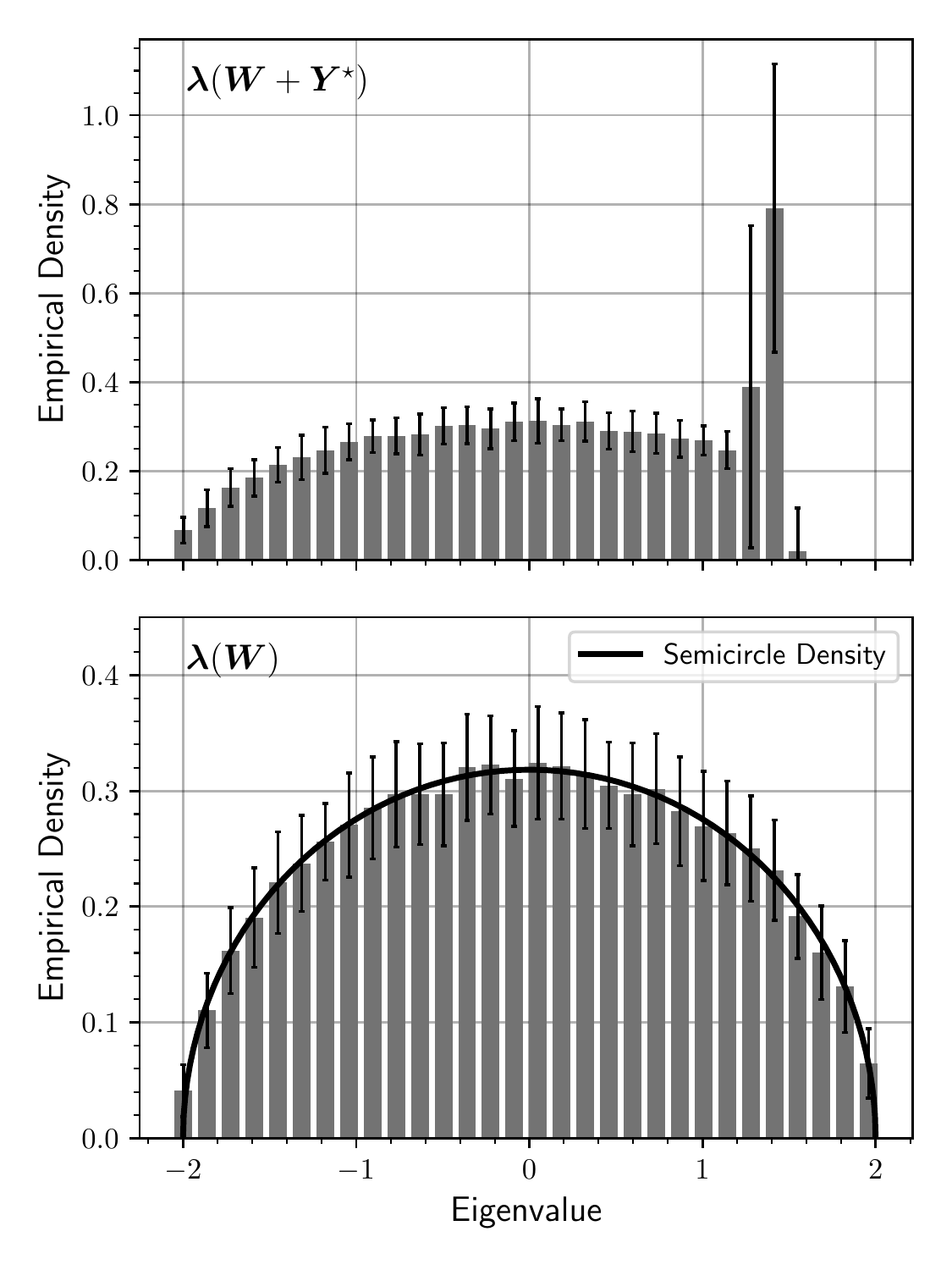}
    \vspace{-1.5em}
    \caption{\textbf{SDP dual spectrum for small $n$.} We plot the means of histograms, with error bars of one standard deviation per bin, for the spectra of $\bW$ and $\bW + \bY^{\star}$ for $\bY^{\star}$ the optimizer of $\SDP^*(\bW)$, over 50 trials with $n = 150$.
    }
    \label{fig:dual-hist}
\end{figure}

\begin{figure}
    \centering
    \includegraphics[scale=0.85]{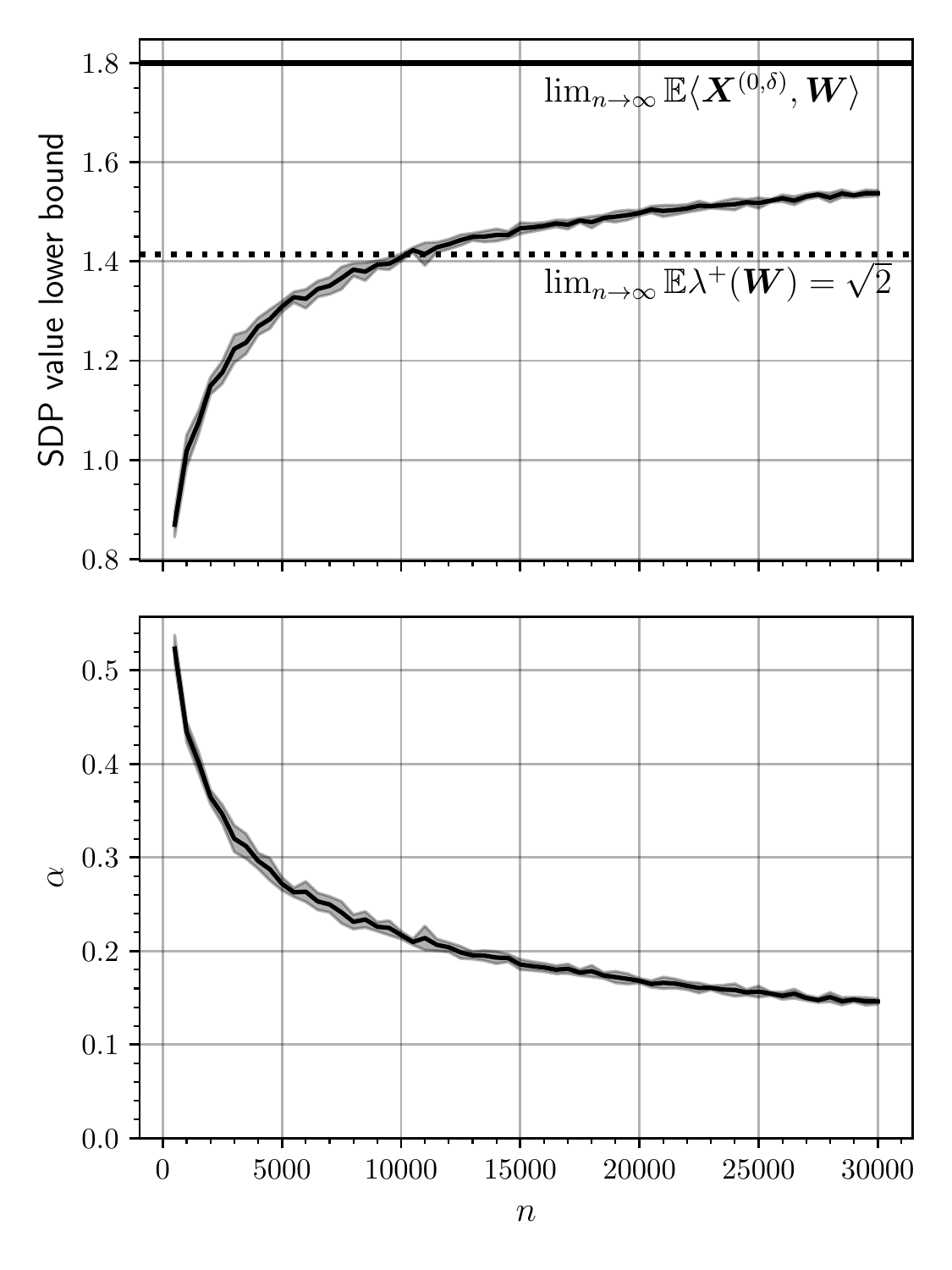}
    \vspace{-1.5em}
    \caption{\textbf{SDP lower bound convergence.} We fix $\delta = 1/25 = 0.04$, and given $\bW$ compute the smallest $\alpha$ for such that $\bX^{(\alpha, \delta)}$, as defined in \eqref{eq:Xad-def}, is feasible for $\SDP(\bW)$. In the upper graph, for a range of values of $n$, we plot the mean and an error interval of one standard deviation of 10 values of $\langle \bX^{(\alpha, \delta)}, \bW \rangle$, a lower bound on $\SDP(\bW)$. We note that this clearly exceeds $\lim_{n \to \infty} \EE \lambda^+(\bW) \approx \sqrt{2}$ once $n \gtrsim 10^4$; however, it only very slowly approaches its expected limiting value. In the bottom graph, we likewise plot the mean and standard deviation of the minimum valid value of $\alpha$, again observing that it only very slowly approaches its limiting value of zero (which may be verified to match the rate our theoretical calculation predicts of $\alpha = O(n^{-1/2})$ up to logarithmic factors).}
    \label{fig:lower-bound}
\end{figure}

\newpage

\bibliographystyle{alpha}
\addcontentsline{toc}{section}{References}
\bibliography{main}

\end{document}